%% file: ms.tex
\documentclass[11pt]{article}
\usepackage{fullpage}
\usepackage{amsmath,amsthm,amssymb}
\usepackage{xcolor}
\usepackage{hyperref}
\usepackage[nameinlink,capitalise]{cleveref}
\usepackage{natbib}
\usepackage{bbold}
\usepackage{mathtools}

\input{macros}

\title{Approximately Efficient Bilateral Trade}

\author{Yuan Deng\thanks{Google Research. Email: \texttt{\{dengyuan,maojm,balusivan\}@google.com}.} \and Jieming Mao\samethanks[1] \and Balasubramanian Sivan\samethanks[1] \and Kangning Wang\thanks{Duke University. Email: \texttt{knwang@cs.duke.edu}.}}

\date{}

\begin{document}
\maketitle

\input{abstract}

\input{intro}

\input{proof}

\bibliographystyle{plainnat}
\bibliography{ref}

\appendix
\input{appendix}

\end{document}

%% file: macros.tex
\makeatletter
\def\@fnsymbol#1{\ensuremath{\ifcase#1\or \dagger\or \ddagger\or
   \mathsection\or \mathparagraph\or \|\or **\or \dagger\dagger
   \or \ddagger\ddagger \else\@ctrerr\fi}}
\makeatother
\newcommand*\samethanks[1][\value{footnote}]{\footnotemark[#1]}

\DeclareMathOperator*{\argmax}{arg\,max}

\DeclareMathOperator*{\E}{\mathrm{E}}

\newcommand{\gft}{\mathsf{GFT}}

\newcommand{\fb}{\mathsf{FB}}
\renewcommand{\sb}{\mathsf{SB}}

\newcommand{\fp}{\mathsf{FixedP}}
\renewcommand{\sp}{\mathsf{SellerP}}
\newcommand{\bp}{\mathsf{BuyerP}}
\renewcommand{\d}{\,\mathrm{d}}

\newtheorem{theorem}{Theorem}
\newtheorem{lemma}{Lemma}
\theoremstyle{definition}
\newtheorem{remark}{Remark}

%% file: abstract.tex
\begin{abstract}
We study bilateral trade between two strategic agents. The celebrated result of Myerson and Satterthwaite 
states that in general, no incentive-compatible, individually rational and weakly budget balanced mechanism can be efficient. I.e., no mechanism with these properties can guarantee a trade whenever buyer value exceeds seller cost. Given this, a natural question is whether there exists a mechanism with these properties that guarantees a \emph{constant fraction} of the first-best gains-from-trade, namely a constant fraction of the gains-from-trade attainable whenever buyer's value weakly exceeds seller's cost. In this work, we positively resolve this long-standing open question on constant-factor approximation, mentioned in several previous works, using a simple mechanism. 
\end{abstract}

%% file: intro.tex
\section{Introduction}
In a bilateral trade, a seller holds an item and is trying to sell it to a buyer. The buyer's private value $v$ is drawn from a cumulative distribution function (CDF) $F$, and the seller's private cost $c$ for selling the item is independently drawn from a CDF $G$. If a transaction happens between $v$ and $c$, the society as a whole gains utility of $v - c$. The \emph{gains-from-trade} ($\gft$) refer to the expected utility gain from trading. If trading probability between $v$ and $c$ is $x(v, c)$, then
\[
\gft = \E_{\substack{v \sim F \\ c \sim G}}\big[(v - c) \cdot x(v, c)\big].
\]
Ideally, to maximize $\gft$, a trade should always happen when $v > c$ and never happen when $v < c$. The resulting optimal $\gft$ is termed the \emph{first best} ($\fb$). Namely,
\[
\fb = \E_{\substack{v \sim F \\ c \sim G}}\big[(v - c) \cdot \mathbb{1}\{v \geq c\}\big].
\]

The seminal work of \citet{myerson1983efficient} shows that if both agents are self-interested, it is impossible to devise a Bayesian incentive-compatible (BIC), individually rational (IR) and weakly budget balanced\footnote{A mechanism is weakly budget balanced if the total payment of the participants is ex-post non-negative, i.e., the payment from the buyer is always at least the revenue of the seller. Moreover, a mechanism is strongly budget balanced if the total payment of the participants is ex-post exactly $0$, i.e., the payment from the buyer is always the revenue of the seller.} (WBB) mechanism that achieves the first-best $\gft$, as long as the distribution supports of $F$ and $G$ ``overlap''. This impossibility result motivates the natural question of whether there exists a mechanism with these properties (BIC, IR, and WBB) that guarantees a \emph{constant fraction} of the first-best $\gft$.

\paragraph{Benchmarks and Simple Mechanisms}
To better explain related work and our results, we first introduce some benchmarks and simple mechanisms. For simplicity, we assume $F$ and $G$ are both continuous distributions supported on a bounded interval $[0, 1]$ with positive densities. We show in \cref{rem:bounded_continuous_dist} why this is without loss of generality. We slightly overload the notations below and use them to denote both the mechanisms and the gains-from-trade they obtain.
\begin{itemize}
\item $\fb$: the first best. It captures the optimal $\gft$ if the agents are not strategic and there is a trade whenever $v \geq c$. Formally, $\fb = \int_0^1 \int_c^1 (v - c) \d F(v) \d G(c)$.
\item $\sb$: the second best. It captures the optimal $\gft$ achieved by any Bayesian incentive-compatible (BIC), individually rational (IR) and weakly budget balanced (WBB) mechanism.
\item $\fp$: the fixed-price mechanism. The mechanism sets a fixed price $p$, and there is a trade if buyer's value is at least $p$ and seller's cost is at most $p$. Formally,  $\fp = \max_p \int_0^p \int_p^1 (v - c) \d F(v) \d G(c)$. The mechanism is dominant-strategy incentive-compatible (DSIC, stronger than BIC), IR and strongly budget balanced (SBB, stronger than WBB).
\item $\sp$: the seller-pricing mechanism. The mechanism delegates the pricing power to the seller, who in turn posts a price $r_c$ to maximize her profit with knowledge of $c$ and $F$. The buyer then decides whether to buy depending on whether $v \geq r_c$. Formally, $\sp = \int_0^1 \int_{r_c}^1 (v - c) \d F(v) \d G(c)$, where $r_c \in \argmax_{p} (p - c) \cdot (1 - F(p))$ is the price the seller sets when she has cost $c$. The mechanism is BIC, IR and SBB.
\item $\bp$: the buyer-pricing mechanism. Symmetric to $\sp$, the mechanism delegates the pricing power to the buyer, who sets a price $r'_v$ to maximize his utility, and then the seller decides whether to sell based on whether $c \leq r'_v$. Formally, $\bp = \int_0^1 \int_0^{r'_v} (v - c) \d G(c) \d F(v)$, where $r'_v \in \argmax_{p} (v - p) \cdot G(p)$. The mechanism is BIC, IR and SBB.
\end{itemize}

With these notations, the result of \citet{myerson1983efficient} demonstrates $\sb < \fb$ whenever the supports of $F$ and $G$ ``overlap''. It has remained an 
open question ever since, on how far apart $\sb$ and $\fb$ can be. Specifically, is it the case that $\sb$ is always at least a constant fraction of $\fb$? We answer this question in the positive. In particular, the better of (or a randomization over) $\sp$ and $\bp$ guarantees at least $10\%$ of the first-best gains-from-trade.

\begin{theorem}
\label{thm:main}
$\fb \leq 2 \cdot \sp + 8 \cdot \bp \leq 10 \cdot \max(\sp, \bp) \leq 10 \cdot \sb$.
\end{theorem}
\begin{remark}
In \cref{sec:improve}, we improve the constant to get a $8.23$-approximation.
\end{remark}

\subsection{Related Work}

\begin{sloppypar}
There has been a large body of work towards answering whether $\sb = \Omega(1) \cdot \fb$. In particular, \citet*{mcafee2008gains} shows $\fp \geq \frac{1}{2} \cdot \fb$ when $F$ and $G$ are i.i.d. and \citet*{DBLP:conf/wine/BlumrosenM16} show $\sp \geq \frac{1}{e} \cdot \fb$ when $F$ satisfies the monotone-hazard-rate condition. As for the negative direction, \citet{DBLP:conf/wine/BlumrosenM16} show that for any $\varepsilon > 0$, $\sb \leq \left(\frac{2}{e} + \varepsilon\right) \cdot \fb$ in some instance. \citet*{DBLP:journals/corr/BlumrosenD16} show that for any $\varepsilon > 0$, there is some instance where $\fp < \varepsilon \cdot \fb$ (and further $\fp < \varepsilon \cdot \sb$).

A closely related question of welfare approximation has also been extensively studied before. The welfare equals $\gft$ plus the term $\E_{c \sim G}[c]$, i.e., the welfare equals buyer's value when trade happens and it equals seller's cost when no trade happens. Maximizing the welfare is equivalent to maximizing the gains-from-trade, but providing a constant approximation to first-best gains-from-trade is much harder than doing the same for first-best welfare. For instance, not having any trade at all already gives a non-zero (and often good) approximation to welfare, but it gives a zero approximation to gains-from-trade.  \citet*{DBLP:journals/corr/BlumrosenD16} show a $\left(1 - \frac{1}{e}\right)$-approximation to the first-best welfare. \citet*{DBLP:journals/corr/abs-2107-14327} improve the approximation ratio to $1 - \frac{1}{e} + 10^{-4}$.

\citet*{DBLP:conf/sigecom/BrustleCWZ17} present a $2$-approximation against second-best gains-from-trade; in particular, they show $\sb \leq \sp + \bp$. Our result is ``stronger'' in the sense that we show the better of $\sp$ and $\bp$ is a constant-approximation of $\fb$ and not only $\sb$. \citet*{DBLP:conf/stoc/DuttingFLLR21} and \citet*{DBLP:journals/corr/abs-2107-14327} consider settings without full knowledge of the distributions.

The \emph{double auction} is a generalized version of bilateral trade, where there are multiple buyers and sellers in the market. There has been a significant amount of work on characterizing and approximating the efficient solutions in related settings; see e.g. \citep*{DBLP:conf/soda/Colini-Baldeschi16,DBLP:conf/sigecom/Colini-Baldeschi17,DBLP:conf/sigecom/BrustleCWZ17,DBLP:conf/sigecom/Babaioff0GZ18,DBLP:conf/soda/BalseiroMLZ19,DBLP:conf/soda/BabaioffGG20,DBLP:conf/soda/CaiGMZ21}.
\end{sloppypar}

%% file: proof.tex
\section{Proof of Constant Approximation}
Recall that $F$ and $G$ are independent continuous CDFs supported on $[0, 1]$ with positive densities (see \cref{rem:bounded_continuous_dist} for why everything apart from independence is without loss of generality). Let $\mu(x)$ be the median of $F_{| \geq x}$ (where $F_{| \geq x}(z) = 0$ for $z < x$ and $F_{| \geq x}(z) = \frac{F(z) - F(x)}{1 - F(x)}$ for $z \geq x$), i.e., $\mu(x) = F^{(-1)} \left(\frac{1 + F(x)}{2}\right)$. Let $\mu^{(k)}(\cdot)$ be the composition of $k$ functions of $\mu(\cdot)$
, and $\mu^{(-k)}(\cdot)$ be its inverse ($0$ if it does not exist). We first give immediate bounds for $\fb$, $\sp$ and $\bp$:
\begin{align*}
\fb &= \int_0^1 \int_c^1 (v - c) \d F(v) \d G(c)\\
&\leq 2 \int_0^1 \int_{\mu(c)}^1 (v - c) \d F(v) \d G(c),
\end{align*}
where the second step holds since $\mu(x)$ is the median of $F_{| \geq x}$, and we preserve the better half.

For $\sp$, we have
\begin{align*}
\sp &\geq \int_0^1 \int_{\mu(c)}^1 \Big(\mu(c) - c\Big) \d F(v) \d G(c).
\end{align*}
The right-hand side (RHS) is the seller's expected profit when setting the price at $\mu(c)$ for each $c$, which is at most her optimal profit and thus at most the gains-from-trade in the profit-optimal seller-pricing mechanism (which is the left-hand side (LHS)).

Similarly, for $\bp$, we have
\begin{align*}
\bp &\geq \int_{\mu^{(2)}(0)}^1 \int_{0}^{\mu^{(-2)}(v)} \Big(v - \mu^{(-2)}(v)\Big) \d G(c) \d F(v)\\
&= \int_{0}^1 \int_{\mu^{(2)}(c)}^1 \Big(v - \mu^{(-2)}(v)\Big) \d F(v) \d G(c).
\end{align*}
Again, the RHS is the buyer's expected utility when setting the price at $\mu^{(-2)}(v)$ for each $v$, which is at most the gains-from-trade in the utility-optimal buyer-pricing mechanism (the LHS).

For each cost $c$, define $\fb(c) = 2 \int_{\mu(c)}^1 (v - c) \d F(v)$, $\sp(c) = \int_{\mu(c)}^1 \Big(\mu(c) - c\Big) \d F(v)$ and $\bp(c) = \int_{\mu^{(2)}(c)}^1 \Big(v - \mu^{(-2)}(v)\Big) \d F(v)$. The bounds above simplify to:
\[
\fb \leq \int_0^1 \fb(c) \d G(c), \quad \sp \geq \int_0^1 \sp(c) \d G(c), \quad \bp \geq \int_0^1 \bp(c) \d G(c).
\]

We now show that, for any $c$, $\fb(c) \leq 2 \cdot \sp(c) + 8 \cdot \bp(c)$, thus proving our result. The crux of our proof is a partition of value space by quantile. We decompose $\fb(c)$ into two parts and bound them by $\sp(c)$ and $\bp(c)$ separately. 

\begin{lemma}
For any $c$, $\fb(c) \leq 2 \cdot \sp(c) + 8 \cdot \bp(c)$.
\end{lemma}
\begin{proof}
We first have
\begin{align*}
\fb(c) &= 2 \int_{\mu(c)}^1 (v - c) \d F(v)\\
&= 2 \int_{\mu(c)}^1 \Big(\mu(c) - c\Big) \d F(v) + 2 \int_{\mu(c)}^1 \Big(v - \mu(c)\Big) \d F(v)\\
&= 2 \cdot \sp(c) + 2 \int_{\mu(c)}^1 \Big(v - \mu(c)\Big) \d F(v).
\end{align*}
We now proceed to show that $\int_{\mu(c)}^1 \Big(v - \mu(c)\Big) \d F(v) \leq 4 \cdot \bp(c)$. Observe that
\begin{align*}
\int_{\mu(c)}^1 \Big(v - \mu(c)\Big) \d F(v) &= \sum_{k = 1}^{\infty} \int_{\mu^{(k)}(c)}^{\mu^{(k + 1)}(c)} \Big(v - \mu(c)\Big) \d F(v)\\
&\leq \sum_{k = 1}^{\infty} \int_{\mu^{(k)}(c)}^{\mu^{(k + 1)}(c)} \Big(\mu^{(k + 1)}(c) - \mu(c)\Big) \d F(v)\\
&= \sum_{k = 1}^{\infty} \sum_{t = 1}^k \int_{\mu^{(k)}(c)}^{\mu^{(k + 1)}(c)} \Big(\mu^{(t + 1)}(c) - \mu^{(t)}(c)\Big) \d F(v)\\
&= \sum_{t = 1}^{\infty} \sum_{k = t}^{\infty} \int_{\mu^{(k)}(c)}^{\mu^{(k + 1)}(c)} \Big(\mu^{(t + 1)}(c) - \mu^{(t)}(c)\Big) \d F(v)\\
&= 2 \sum_{t = 1}^{\infty} \int_{\mu^{(t)}(c)}^{\mu^{(t + 1)}(c)} \Big(\mu^{(t + 1)}(c) - \mu^{(t)}(c)\Big) \d F(v).
\end{align*}
The last step uses the definition of $\mu(\cdot)$: $k = t$ counts for half the probability of all $k \geq t$.

On the other hand, for $\bp(c)$, recall that $\bp(c) = \int_{\mu^{(2)}(c)}^1 \Big(v - \mu^{(-2)}(v)\Big) \d F(v)$ and we have
\begin{align*}
\int_{\mu^{(2)}(c)}^1 \Big(v - \mu^{(-2)}(v)\Big) \d F(v) &= \sum_{t = 2}^{\infty} \int_{\mu^{(t)}(c)}^{\mu^{(t + 1)}(c)} \Big(v - \mu^{(-2)}(v)\Big) \d F(v)\\
&\geq \sum_{t = 2}^{\infty} \int_{\mu^{(t)}(c)}^{\mu^{(t + 1)}(c)} \Big(\mu^{(t)}(c) - \mu^{(t - 1)}(c)\Big) \d F(v)\\
&= \sum_{t = 1}^{\infty} \int_{\mu^{(t + 1)}(c)}^{\mu^{(t + 2)}(c)} \Big(\mu^{(t + 1)}(c) - \mu^{(t)}(c)\Big) \d F(v)\\
&= \frac{1}{2} \sum_{t = 1}^{\infty} \int_{\mu^{(t)}(c)}^{\mu^{(t + 1)}(c)} \Big(\mu^{(t + 1)}(c) - \mu^{(t)}(c)\Big) \d F(v),
\end{align*}
where the last step again applies the definition of $\mu(\cdot)$. Therefore, we conclude the proof that $\fb(c) \leq 2 \cdot \sp(c) + 8 \cdot \bp(c)$.
\end{proof}

\begin{remark}\label{rem:bounded_continuous_dist}
It is without loss of generality to consider bounded continuous distributions with positive densities: For any general distributions $F$ and $G$ on $[0, 1]$, there is a sequence of continuous distributions on $[0, 1]$ with positive densities converging to it in the L\'evy metric. All of $\fb$, seller's profit in $\sp$ and buyer's utility in $\bp$ are continuous in $F$ and $G$ (in the L\'evy metric). It is also without loss of generality to consider bounded supports: As long as $\fb$ is finite, almost all of the contribution comes from $(v, c) \in [-M, M]^2$ for some $M$.
\end{remark}

\begin{remark}
By symmetry, we also have $\fb \leq 8 \cdot \sp + 2 \cdot \bp$. Thus, using $\sp$ with probability $\alpha$ and $\bp$ with probability $1 - \alpha$ for any $\alpha \in [0.2, 0.8]$ gives a $10$-approximation to $\fb$.
\end{remark}

\begin{remark}
We can use a parameter to control the quantile of $\mu(x)$ (instead of using the median) in $F_{| \geq x}$. This improves the approximation constant to $8.23$. The details are deferred to \cref{sec:improve}.
\end{remark}

%% file: appendix.tex
\section{Improving the Constant}
\label{sec:improve}
To improve the approximation constant of \cref{thm:main}, instead of setting $\mu(x)$ to be the median of $F_{| \geq x}$, we introduce a parameter $\lambda$ to control the quantile of $\mu(x)$ in $F_{| \geq x}$:
\[
\mu(x) := F^{(-1)} \Big(\lambda + (1 - \lambda) F(x)\Big).
\]

Similar to the proof of \cref{thm:main}, we have
\begin{align*}
\fb &= \int_0^1 \int_c^1 (v - c) \d F(v) \d G(c) \leq \frac{1}{1 - \lambda} \int_0^1 \int_{\mu(c)}^1 (v - c) \d F(v) \d G(c),
\end{align*}
\begin{align*}
\sp &\geq \int_0^1 \int_{\mu(c)}^1 \Big(\mu(c) - c\Big) \d F(v) \d G(c),
\end{align*}
and
\begin{align*}
\bp &\geq \int_{\mu^{(2)}(0)}^1 \int_{0}^{\mu^{(-2)}(v)} \left(v - \mu^{(-2)}(v)\right) \d G(c) \d F(v) = \int_{0}^1 \int_{\mu^{(2)}(c)}^1 \left(v - \mu^{(-2)}(v)\right) \d F(v) \d G(c).
\end{align*}

Let $\fb(c) = \frac{1}{1 - \lambda} \int_{\mu(c)}^1 (v - c) \d F(v)$, $\sp(c) = \int_{\mu(c)}^1 (\mu(c) - c) \d F(v)$ and $\bp(c) = \int_{\mu^{(2)}(c)}^1 (v - \mu^{(-2)}(v)) \d F(v)$. We have
\[
\fb \leq \int_0^1 \fb(c) \d G(c), \quad \sp \geq \int_0^1 \sp(c) \d G(c), \quad \bp \geq \int_0^1 \bp(c) \d G(c).
\]

Fix any $c$ from now on.
\begin{align*}
\fb(c) &= \frac{1}{1 - \lambda} \int_{\mu(c)}^1 (v - c) \d F(v)\\
&= \frac{1}{1 - \lambda} \int_{\mu(c)}^1 \Big(\mu(c) - c\Big) \d F(v) + \frac{1}{1 - \lambda} \int_{\mu(c)}^1 \Big(v - \mu(c)\Big) \d F(v)\\
&= \frac{1}{1 - \lambda} \cdot \sp(c) + \frac{1}{1 - \lambda} \sum_{k = 1}^{\infty} \int_{\mu^{(k)}(c)}^{\mu^{(k + 1)}(c)} \Big(v - \mu(c)\Big) \d F(v)\\
&\leq \frac{1}{1 - \lambda} \cdot \sp(c) + \frac{1}{1 - \lambda} \sum_{k = 1}^{\infty} \int_{\mu^{(k)}(c)}^{\mu^{(k + 1)}(c)} \left(\mu^{(k + 1)}(c) - \mu(c)\right) \d F(v)\\
&= \frac{1}{1 - \lambda} \cdot \sp(c) + \frac{1}{1 - \lambda} \sum_{k = 1}^{\infty} \sum_{t = 1}^k \int_{\mu^{(k)}(c)}^{\mu^{(k + 1)}(c)} \left(\mu^{(t + 1)}(c) - \mu^{(t)}(c)\right) \d F(v)\\
&= \frac{1}{1 - \lambda} \cdot \sp(c) + \frac{1}{1 - \lambda} \sum_{t = 1}^{\infty} \sum_{k = t}^{\infty} \int_{\mu^{(k)}(c)}^{\mu^{(k + 1)}(c)} \left(\mu^{(t + 1)}(c) - \mu^{(t)}(c)\right) \d F(v)\\
&= \frac{1}{1 - \lambda} \cdot \sp(c) + \frac{1}{\lambda (1 - \lambda)} \sum_{t = 1}^{\infty} \int_{\mu^{(t)}(c)}^{\mu^{(t + 1)}(c)} \left(\mu^{(t + 1)}(c) - \mu^{(t)}(c)\right) \d F(v).
\end{align*}
The last step uses the definition of $\mu(\cdot)$: $k = t$ counts for $\lambda$-fraction of the probability of all $k \geq t$.

Additionally, we have
\begin{align*}
\bp(c) &= \int_{\mu^{(2)}(c)}^1 \left(v - \mu^{(-2)}(v)\right) \d F(v)\\
&= \sum_{t = 2}^{\infty} \int_{\mu^{(t)}(c)}^{\mu^{(t + 1)}(c)} \left(v - \mu^{(-2)}(v)\right) \d F(v)\\
&\geq \sum_{t = 2}^{\infty} \int_{\mu^{(t)}(c)}^{\mu^{(t + 1)}(c)} \left(\mu^{(t)}(c) - \mu^{(t - 1)}(c)\right) \d F(v)\\
&= \sum_{t = 1}^{\infty} \int_{\mu^{(t + 1)}(c)}^{\mu^{(t + 2)}(c)} \left(\mu^{(t + 1)}(c) - \mu^{(t)}(c)\right) \d F(v)\\
&= (1 - \lambda) \sum_{t = 1}^{\infty} \int_{\mu^{(t)}(c)}^{\mu^{(t + 1)}(c)} \left(\mu^{(t + 1)}(c) - \mu^{(t)}(c)\right) \d F(v).
\end{align*}

Therefore, $\fb(c) \leq \frac{1}{1 - \lambda} \cdot \sp(c) + \frac{1}{\lambda (1 - \lambda)^2} \cdot \bp(c)$. Setting $\lambda = 0.311$ gives $\fb \leq 8.23 \cdot \max(\sp, \bp)$.

%% file: ms.bbl
\begin{thebibliography}{13}
\providecommand{\natexlab}[1]{#1}
\providecommand{\url}[1]{\texttt{#1}}
\expandafter\ifx\csname urlstyle\endcsname\relax
  \providecommand{\doi}[1]{doi: #1}\else
  \providecommand{\doi}{doi: \begingroup \urlstyle{rm}\Url}\fi

\bibitem[Babaioff et~al.(2018)Babaioff, Cai, Gonczarowski, and
  Zhao]{DBLP:conf/sigecom/Babaioff0GZ18}
Moshe Babaioff, Yang Cai, Yannai~A. Gonczarowski, and Mingfei Zhao.
\newblock The best of both worlds: Asymptotically efficient mechanisms with a
  guarantee on the expected gains-from-trade.
\newblock In \emph{EC}, page 373, 2018.

\bibitem[Babaioff et~al.(2020)Babaioff, Goldner, and
  Gonczarowski]{DBLP:conf/soda/BabaioffGG20}
Moshe Babaioff, Kira Goldner, and Yannai~A. Gonczarowski.
\newblock Bulow-klemperer-style results for welfare maximization in two-sided
  markets.
\newblock In \emph{SODA}, pages 2452--2471, 2020.

\bibitem[Balseiro et~al.(2019)Balseiro, Mirrokni, Paes~Leme, and
  Zuo]{DBLP:conf/soda/BalseiroMLZ19}
Santiago~R. Balseiro, Vahab~S. Mirrokni, Renato Paes~Leme, and Song Zuo.
\newblock Dynamic double auctions: towards first best.
\newblock In \emph{SODA}, pages 157--172, 2019.

\bibitem[Blumrosen and Dobzinski(2016)]{DBLP:journals/corr/BlumrosenD16}
Liad Blumrosen and Shahar Dobzinski.
\newblock ({A}lmost) efficient mechanisms for bilateral trading.
\newblock \emph{CoRR}, abs/1604.04876, 2016.

\bibitem[Blumrosen and Mizrahi(2016)]{DBLP:conf/wine/BlumrosenM16}
Liad Blumrosen and Yehonatan Mizrahi.
\newblock Approximating gains-from-trade in bilateral trading.
\newblock In \emph{WINE}, pages 400--413, 2016.

\bibitem[Brustle et~al.(2017)Brustle, Cai, Wu, and
  Zhao]{DBLP:conf/sigecom/BrustleCWZ17}
Johannes Brustle, Yang Cai, Fa~Wu, and Mingfei Zhao.
\newblock Approximating gains from trade in two-sided markets via simple
  mechanisms.
\newblock In \emph{EC}, pages 589--590, 2017.

\bibitem[Cai et~al.(2021)Cai, Goldner, Ma, and Zhao]{DBLP:conf/soda/CaiGMZ21}
Yang Cai, Kira Goldner, Steven Ma, and Mingfei Zhao.
\newblock On multi-dimensional gains from trade maximization.
\newblock In \emph{SODA}, pages 1079--1098, 2021.

\bibitem[Colini{-}Baldeschi et~al.(2016)Colini{-}Baldeschi, de~Keijzer,
  Leonardi, and Turchetta]{DBLP:conf/soda/Colini-Baldeschi16}
Riccardo Colini{-}Baldeschi, Bart de~Keijzer, Stefano Leonardi, and Stefano
  Turchetta.
\newblock Approximately efficient double auctions with strong budget balance.
\newblock In \emph{SODA}, pages 1424--1443, 2016.

\bibitem[Colini{-}Baldeschi et~al.(2017)Colini{-}Baldeschi, Goldberg,
  de~Keijzer, Leonardi, Roughgarden, and
  Turchetta]{DBLP:conf/sigecom/Colini-Baldeschi17}
Riccardo Colini{-}Baldeschi, Paul~W. Goldberg, Bart de~Keijzer, Stefano
  Leonardi, Tim Roughgarden, and Stefano Turchetta.
\newblock Approximately efficient two-sided combinatorial auctions.
\newblock In \emph{EC}, pages 591--608, 2017.

\bibitem[D{\"{u}}tting et~al.(2021)D{\"{u}}tting, Fusco, Lazos, Leonardi, and
  Reiffenh{\"{a}}user]{DBLP:conf/stoc/DuttingFLLR21}
Paul D{\"{u}}tting, Federico Fusco, Philip Lazos, Stefano Leonardi, and Rebecca
  Reiffenh{\"{a}}user.
\newblock Efficient two-sided markets with limited information.
\newblock In Samir Khuller and Virginia~Vassilevska Williams, editors,
  \emph{STOC}, pages 1452--1465, 2021.

\bibitem[Kang et~al.(2021)Kang, Pernice, and
  Vondr{\'{a}}k]{DBLP:journals/corr/abs-2107-14327}
Zi~Yang Kang, Francisco Pernice, and Jan Vondr{\'{a}}k.
\newblock Fixed-price approximations in bilateral trade.
\newblock \emph{CoRR}, abs/2107.14327, 2021.

\bibitem[McAfee(2008)]{mcafee2008gains}
R~Preston McAfee.
\newblock The gains from trade under fixed price mechanisms.
\newblock \emph{Applied economics research bulletin}, 1\penalty0 (1):\penalty0
  1--10, 2008.

\bibitem[Myerson and Satterthwaite(1983)]{myerson1983efficient}
Roger~B. Myerson and Mark~A. Satterthwaite.
\newblock Efficient mechanisms for bilateral trading.
\newblock \emph{Journal of economic theory}, 29\penalty0 (2):\penalty0
  265--281, 1983.

\end{thebibliography}
